\documentclass[letterpaper]{article} 
\usepackage{times}  
\usepackage{helvet}  
\usepackage{courier}  
\usepackage{url}  
\usepackage{graphicx}  
\usepackage{amsmath,amssymb}
\usepackage{bbm}
\usepackage{graphicx}
\usepackage{mathrsfs}
\usepackage{tikz}
\usepackage{booktabs}
\usetikzlibrary{bayesnet}
\usepackage{amsthm}
\newtheorem{definition}{Definition}

\newtheorem{prop}{Proposition}

\title{A Probabilistic Model of the Bitcoin Blockchain}
\author{Marc Jourdan, Sebastien Blandin, Laura Wynter, Pralhad Deshpande\\
marc.jourdan@polytechnique.edu; \{sblandin, lwynter, pralhad\}@sg.ibm.com\\
IBM Research\\
10 Marina Boulevard \\
18983 Singapore \\
}
 
\begin{document}
\maketitle
\begin{abstract}
The Bitcoin transaction graph is a public data structure organized as transactions between addresses, each associated with a logical entity. In this work, we introduce a complete probabilistic model of the Bitcoin Blockchain. We first formulate a set of conditional dependencies induced by the Bitcoin protocol at the block level and derive a corresponding fully observed graphical model of a Bitcoin block. We then extend the model to include hidden entity attributes such as the functional category of the associated logical agent and derive asymptotic bounds on the privacy properties implied by this model. At the network level, we show evidence of complex transaction-to-transaction behavior and present a relevant discriminative model of the agent categories. Performance of both the block-based graphical model and the network-level discriminative model is evaluated on a subset of the public Bitcoin Blockchain.
\end{abstract}    
\begin{section}{Introduction}
Analysis of the Bitcoin Blockchain~\cite{Nakamoto} is an area of intense activity~\cite{FirstFourYear,GraphPrimerBlockchain}, and one which has witnessed an explosion of interest as the value of the Bitcoin cryptocurrency has skyrocketed. Research areas include explorations of address clustering techniques to identify logical agents~\cite{7816867,7796940,7816867,AutomaticBtcClustering}, de-anonymization using side-channel attacks~\cite{NIPS2017_6735,TorDeanonimization}. 

An understanding of the properties of Bitcoin transactions is paramount to the legitimation of the cryptocurrency economy; it constitutes a building block to the conception of effective and adequate regulations~\cite{MoneyLaundering}, and to the design of novel and integrated services benefiting society as a whole.

As of $2018$, with more than $500$ million address nodes, the Bitcoin graph is comparable in size to a large social network. Yet while probabilistic models of social networks have received considerable attention, from community detection~\cite{leskovec2009community} to diffusion models and influence maximization~\cite{wilder122018maximizing}, to probabilistic graph modeling~\cite{kuter2007sunny}, probabilistic models of the Bitcoin Blockchain network have not.  

Bitcoin transactions are tantamount to a partially observed social network, within which participants can have multiple seemingly independent aliases. This distinguishes our work from classical studies on partially observed social networks, typically focused on partial observations of interactions due to sampling~\cite{handcock2010modeling}, and makes it closer to the vast body of work on entity resolution~\cite{singla2006entity,bhattacharya2006latent}.

A second challenge associated with modeling the Bitcoin Blockchain transaction network consists of capturing the complexity of the hidden structure associated with entity transactions, together with the fine-grained block-level specificities implied by the Bitcoin protocol. In particular, Bitcoin is based on an \textit{unspent transaction output} (UTXO) model, which distinguishes suitable Bitcoin Blockchain models from prior studies on credit card transactions~\cite{dal2015credit,lebichot2016graph}, since the proper generative structure needs to account for the underlying UTXO creation and deletion process.

In this work, we propose a first attempt at a comprehensive model of the Bitcoin transaction graph using a hybrid generative-discriminative model attempting to draw strengths from both approaches~\cite{ng2002discriminative}. We first define pragmatic conditional independence assumptions underlying the Bitcoin protocol, and formulate a generative model of the Bitcoin Blockchain block. In this context, we analyze the revealed entity behavior, both theoretically and from a data perspective. We then turn to network level modeling, present a discriminative model of transaction-transaction behavior, and analyze the associated medium-term categorical agent behavior.
\end{section}

\begin{section}{Probabilistic block model}~\label{sec:graph}
A Bitcoin transaction consists of a set of input addresses transferring BTC to a set of output addresses. More specifically, in the context of a transaction, each input address contributes a possibly fractional subset of its UTXOs to the creation of the set of UTXOs associated with output addresses, for the same total amount (minus a fee). Each UTXO is associated with an address, and each address is associated with a  logical agent, who may hold an arbitrary number of addresses, see Figure~\ref{fig:EntityGraph}.
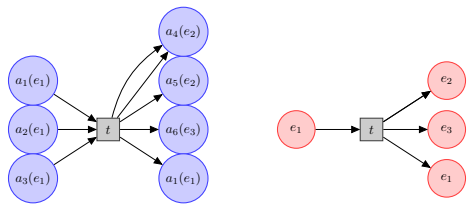
\begin{figure}[htb!]
\centering
\scalebox{0.50}{\begin{tikzpicture}[node distance=1.3cm,bend angle=20,auto]

  \tikzstyle{place}=[circle,thick,draw=blue!75,fill=blue!20,minimum size=10mm]
  \tikzstyle{redplace}=[circle,thick,draw=red!75,fill=red!20,minimum size=10mm]
  \tikzstyle{transition}=[rectangle,thick,draw=black!75,fill=black!20,minimum size=6mm]
        \begin{scope}
    \node [place] (n1) at (0,6.5) {$a_{2}(e_{1})$};
    \node [place] (n2) [above of=n1]  {$a_{1}(e_{1})$};
    \node [place] (n3) [below of=n1]  {$a_{3}(e_{1})$};
    
    \node [place] (n4) at (4,6.5) {$a_{6}(e_{3})$};
    \node [place] (n5) [above of=n4]  {$a_{5}(e_{2})$};
    \node [place] (n6) [above of=n5]  {$a_{4}(e_{2})$};
    \node [place] (n7) [below of=n4]  {$a_{1}(e_{1})$};
    
    \node [transition] (e1) at (2,6.5) {$t$};
      
    \draw[->]
        (n1) edge (e1) (e1) edge (n4) (n2) edge (e1) (n3) edge (e1) (e1) edge (n5) (e1) edge (n6) (e1) edge (n7);
    \draw[->] (e1) edge[bend left] (n6);
    
  \end{scope}
  
  \begin{scope}[xshift=7cm]
    \node [redplace] (n1') at (0,6.5) {$e_{1}$};
    
    \node [redplace] (n4') at (4,6.5) {$e_{3}$};
    \node [redplace] (n5') [above of=n4']  {$e_{2}$};
        \node [redplace] (n6') [below of=n4']  {$e_{1}$};
    
    \node [transition] (e1') at (2,6.5) {$t$};
      
    \draw[->]  (n1') edge (e1') (e1') edge (n4') (e1') edge (n5') (e1') edge (n5') (e1') edge (n6');
        
  \end{scope}
 
\end{tikzpicture}}
\caption{\textbf{Bitcoin address-transaction} bipartite graph data structure of visible addresses $a_{i}$, (associated with unknown entities $e_{i}$), (left) with each arrow corresponding to an \textit{unspent transaction output} (UTXO), and corresponding hidden entity-transaction graph (right). A block consists of many such transactions. A change address, here $a_{1}$, may be used to return the remainder of the UTXO.} 
\label{fig:EntityGraph}
\end{figure}

We embed the Bitcoin Blockchain transaction graph in a directed bipartite graph structure $\mathscr{G} = (\mathscr{A},\mathscr{T},\mathscr{E})$, with the following vertex and edge features:
\begin{itemize}
\item \textit{address vertex $a \in \mathscr{A}$}: number of UTXO $k_{a}^{UTXO}$, and out-degree $k_{a}^{out}$,
\item \textit{transaction vertex $t \in \mathscr{T}$}: transaction value $v$ and fee $f$, 
\item \textit{directed address-transaction edge $\in \mathscr{E}$}: outgoing value $v$ from address $a$ via transaction $t$,
\item \textit{directed transaction-address edge $\in \mathscr{E}$}: incoming value $v$ to address $a$ via transaction $t$.
\end{itemize}
Since the Bitcoin protocol specifies that transactions should be validated in blocks and the \textit{proof-of-work} consensus protocol incentivizes validators to agree on a single block-chain, we ignore transient disagreements and assume a discrete-time \textit{simple path} structure of blocks.

We propose a stationary graphical model~\cite{jordan2004} of a Bitcoin Blockchain block. First we develop a fully observable \textit{block-transaction, address} (BT-A) model, illustrated in Figure~\ref{fig:plate_notation_simple_bta_model}, that we then augment with entity attributes into a \textit{block-transaction, entity-address} (BT-EA) model with more complex structure.
\begin{figure}[!htb]
    \centering
    \scalebox{0.65}{\begin{tikzpicture}
    
    \node[latent]           (p_utxo_out)        {$p_{UTXO,out}$};%
    \node[latent, right=0.2cm of p_utxo_out]           (p_new_c)        {$p_{new}$};%
    \node[latent, right=0.2cm of p_new_c]           (lam_out_c)        {$\lambda_{out}$};%
    \node[latent, right=0.2cm of lam_out_c]           (lam_in_c)        {$\lambda_{in}$};%
    \node[latent, right=0.2cm of lam_in_c]           (p_utxo_in)        {$p_{UTXO,in}$};%
    
    \node[obs, above=2cm of p_utxo_out]           (U_o)        {$U_{o}$};%
    \node[obs, above=0.7cm of U_o]              (V_o)      {$V_{o,u}$};%
    \node[obs, left=0.3cm of V_o]           (A_o)        {$A_{o}$};%

    \node[obs, above=2cm of p_new_c]           (N_new)        {$N_{t}$};%
    \node[obs, above=2cm of lam_out_c]           (O_t)        {$O_{t}$};%
    \node[obs, above=2.1cm of lam_in_c]           (I_t)        {$I_{t}$};%
    
    \node[obs, right=0.5cm of V_o]           (V_t)        {$V_{t}$};%
    \node[obs, above=0.5cm of V_t]           (F_t)        {$F_{t}$};%
    \node[latent, above=0.5cm of F_t]           (p_fee)        {$\mu_{fee}$};%
    
    \node[obs, above=2cm of p_utxo_in]           (U_i)        {$U_{i}$};%
    \node[obs, above=0.8cm of U_i]              (V_i)      {$V_{i,u}$};%
    \node[obs, right=0.3cm of V_i]           (A_i)        {$A_{i}$};%
    
    \node[latent, right=1.2cm of p_utxo_in]           (lam_size)        {$\lambda_{size}$};%

    \node[obs, left=0.5cm of A_o]           (d_out_deg)        {$k^{out}$};%
    \node[obs, right=0.5cm of A_i]           (d_utxo_available)        {$k^{UTXO}$};%
    \node[obs, above=1cm of lam_size]           (T_b)        {$T_{b}$};%

    \draw[->] (d_out_deg) edge (A_o) %
              (d_utxo_available) edge (A_i) %
              (lam_size) edge (T_b) %
              (lam_in_c) edge (I_t) %
              (p_utxo_in) edge (U_i) %
              (p_utxo_out) edge (U_o) %
              (lam_out_c) edge (O_t) %
              (O_t) edge (N_new) %
              (p_new_c) edge (N_new) %
              (p_fee) edge (F_t) %
              (F_t) edge (V_t) %
              (V_t) edge (V_o) %
              (F_t) edge (V_o) %
              (A_i) edge (V_i) %
              (A_o) edge (V_o) %
              (V_i) edge (V_t)
              (N_new) edge (A_o);

    \plate {UTXO_out} { 
        (V_o)
    } {$U_{o}$}; %
    \plate {out} { 
     (UTXO_out) %
     (U_o)(A_o)
    } {$O_{t}$}; %

    \plate {UTXO_in} { 
        (V_i)
    } {$U_{i}$}; %
    \plate {in} { 
     (UTXO_in) %
     (U_i)(A_i)
    } {$I_{t}$}; %

    \plate {Transaction} { 
     (out) %
     (in) %
     (O_t)(N_new)(I_t)(F_t)(V_t) %
    } {$T_{b}+1$}; %
    
    \plate {Block} { 
     (Transaction) %
     (T_b)(d_out_deg)(d_utxo_available)%
    } {$ B$}; %
\end{tikzpicture}}
    \caption{\textbf{Block-transaction address} model, plate notation. Observed random variables are shaded while non-observed variables are plain.}
    \label{fig:plate_notation_simple_bta_model}
\end{figure}
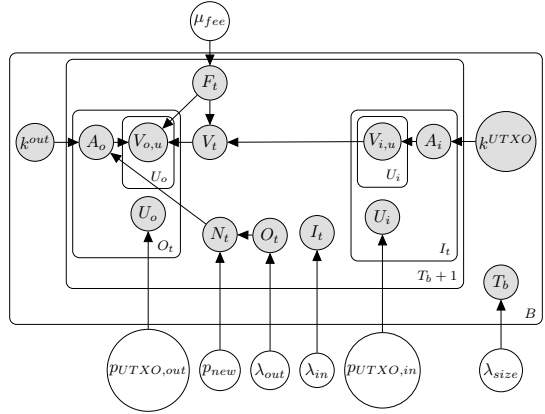
\begin{subsection}{Block-transaction address model}
A block $b$ is composed of the set of transactions $t$ validated by the peer node who solved the cryptographic challenge the fastest. With the approximation of stationary inter-block time, and assuming independence between the ability of solving the cryptographic challenge and the selection of transactions, we model the number $T_{b}$ of transactions per block as a Poisson distribution. Similarly assuming stationary and independent address usage, we model the number of input addresses $I_{t}$ and output addresses $O_{t}$ per transaction as a Poisson random variable.
\begin{definition}[Transactions, input and output addresses]
\begin{align}
& \forall b \in \mathscr{B}, \quad T_{b} \sim \mathscr{P}(\lambda_{size}) \nonumber \\
& \forall t \in \mathscr{T}, \quad I_{t} \sim \mathscr{P}_{n}(\lambda_{in}) \quad \text{and} \quad O_{t} \sim \mathscr{P}_{n}(\lambda_{out}),
\label{eq:transInOut}
\end{align}
where $\mathscr{P}_{n}$ is the normalized \textit{Poisson} distribution on $\mathbb{N}^{*}$.
\end{definition}
On the receiving end of a transaction, it is possible to generate a new address. In the Bitcoin pseudonymous context, this reduces the traceability of the full set of transactions associated with an entity. Considering the set of output addresses as a whole, we model the conditional distribution of the number of new addresses given the number of output addresses as a Binomial random variable.
\begin{definition}[New address distribution]
\begin{equation}
\forall t \in \mathscr{T}, \quad N_{t} | O_{t} \sim \mathscr{B}(O_{t}, p_{new}).
\label{eq:mew}
\end{equation}
\end{definition}
In the interest of a tractable inference procedure, and in the absence of an informative prior, in this work we focus our efforts on maximum-likelihood estimation, and assume uniform prior $\lambda_{size}, \lambda_{in}, \lambda_{out}, p_{new}$.

We now proceed to describe the generative model of the input and output addresses. A natural choice for the generative hierarchical model is the LDA or Dirichlet-multinomial model used in topic modeling~\cite{blei2003latent,mimno2012topic}. Here, given the full observability of the model variables and decomposability of the likelihood, motivated by topological social network analysis, we use the Albert-Barabasi \textit{preferential attachment} model~\cite{AlbertBarabasiScaling}, which can be seen as the posterior probability of an LDA model in the appropriate feature space. 

Specifically, we consider that the probability of the $i^{th}$ address $A_{i}$ to be a given address $a$ is proportional to the number of available UTXO of the address. The model reads as follows.
\begin{definition}[Input addresses]
$\forall i \in \{1, \dots, I_{t}\}$
\begin{equation}
\mathbb{P}(A_{i}=a | k^{UTXO}) = \frac{k_{a}^{UTXO}+1}{\sum_{a' \in \mathscr{A}} (k_{a'}^{UTXO}+1)}
\label{eq:in}
\end{equation}
where $\mathscr{A}$ is the set of available addresses, $k_{a}^{UTXO}$ is the number of unspent outputs of the address.
\end{definition}
The output address model is similar, except that the attachment model is now considered a function of the out-degree of the address, i.e. while the inclination of the address to be part of the inputs (i.e. to spend) is considered to be a function of the number of UTXOs it has still available, the inclination of the address to be part of the outputs (i.e. to accumulate) is considered to be a function of the number of distinct UTXOs it has already spent.
\begin{definition}[Output addresses]
$\forall o \in \{1, \dots, O_{t}\}$
\begin{align}
\mathbb{P}(A_{o}  =a | N_t, k^{out}) \sim & \mathbbm{1}(o \leq N_t; a = a_{0}) \label{eq:out}\\
& + \mathbbm{1} ( o > N_t) \frac{k_{a}^{out}+1}{\sum_{a' \in \mathscr{A}} (k_{a'}^{out}+1)}\nonumber 
\end{align}
where $a_{0}$ denotes a new address.
\end{definition}
For each input address, since empirically we observe that the $UTXO$ distribution is concentrated around $1$, we model the conditional distribution of the number $U_{i}$ of UTXOs used given the number of UTXOs available as a geometric random variable with uniform prior. We then draw the UTXOs uniformly from the available set.
\begin{definition}[Input UTXOs]
\begin{align}
\forall i \in \{1, \dots, I_{t}\}, \, & U_{i} | k_{A_{i}}^{UTXO} \sim \mathscr{G}_{[1,\dots,k_{A_{i}}^{UTXO}]}(p_{UTXO,in}) \nonumber \\
\forall u \in \{1, \dots, U_{i}\}, \, & V_{i,u} | U_{i} \sim \mathscr{U}_{\{1,\dots,k_{A_{i}}^{UTXO}\}}
\label{eq:utxoIn}
\end{align}
where $\mathscr{G}_{[1,\dots,k_{A_{i}}^{UTXO}]}$ is the normalized geometric distribution with support $[1,\dots,k_{A_{i}}^{UTXO}]$, and where $\mathscr{U}_{\{1,\dots,k_{A_{i}}^{UTXO}\}}$ is the uniform distribution over the set ${\{1,\dots,k_{A_{i}}^{UTXO}\}}$.
\end{definition}
We obtain the total transaction value $V_t$ as the sum of the input UTXOs.
\begin{definition}[Transaction value]
\begin{equation}
V_{t} | I_{t},U_{i}, V_{i,u} = \sum_{1 \leq i \leq I_{t}} \sum_{1 \leq u \leq U_{i}} V_{i,u}.
\label{eq:value}
\end{equation}
\end{definition}
A fee is paid to the miners to reward their validation work and higher fees may nudge their selection of transactions when creating blocks. We thus model the fee associated with a Bitcoin transaction as a normalized Gaussian distribution. The number of output UTXOs and their values is modeled similarly to the input UTXOs.
\begin{definition}[Fee value, output UTXOs]
\begin{align}
& \forall t \in \mathscr{T}, \quad F_{t} | V_{t} = \mathscr{N}_{[0,V_{t}]}(\mu_{fee},\sigma_{fee}) \nonumber \\
& \forall o \in \{1, \dots, O_{t}\}, \quad U_{o} \sim \mathscr{G}(p_{UTXO,out}) \nonumber \\
& \forall u \in \{1, \dots, U_{o}\}, \quad V_{o,u} | V_{t}, F_{t} \sim \mathscr{U}_{[1,\dots,V_{t}-F_{t}]}
\label{eq:utxoOut}
\end{align}
where $\mathscr{N}_{[a,b]}$ denotes the Gaussian distribution normalized over the interval $[a,b]$, and where $\mathscr{U}$ denotes the normalized uniform distribution (the $V_{o,u}$ are also normalized in order to sum to $V(t)-F(t)$).
\end{definition}
The resulting block-transaction address model~\eqref{eq:transInOut}-\eqref{eq:mew}-\eqref{eq:in}-\eqref{eq:out}-\eqref{eq:utxoIn}-\eqref{eq:value}-\eqref{eq:utxoOut} is presented in Figure~\ref{fig:plate_notation_simple_bta_model}. 

We now turn to a more complex variant of the proposed model meant to capture categorical behavior of the unobserved entities transacting on the Blockchain.
\end{subsection}

\begin{subsection}{Block-transaction entity-address model}
An entity $e$ is associated with a Bitcoin user and fully characterized by a set of addresses $A(e) = \{a^{(e)}_{i}\}_{i}$. In this section we extend the BT-A model to take into account categorical entity behavior. We assume that entities belong to different categories $c \in \mathscr{C}$, with potentially different behaviors.

We first model the fact that the hyper-parameters $\lambda_{in}$ and $\lambda_{out}$ associated with the number of input and output addresses, depend on the category $c$ of the associated entity, and are noted $\lambda_{in,c}$ and $\lambda_{out,c}$. Similarly the parameter associated with the number of new addresses in the output $p_{new,c}$, and the number of UTXO in the input $p_{UTXO,in,c}$ and output $p_{UTXO,out,c}$ are category-dependent.

Second we update the conditional independence structure of the generative model to reflect the fact that address selection~\eqref{eq:in}-\eqref{eq:out} is now also conditioned on entities.
\begin{definition}[Input and output entities and addresses]
\begin{align}
& \mathbb{P}(E_{t}=e | k^{UTXO}) = \frac{k_{e}^{UTXO}+1}{\sum_{e' \in \mathscr{E}} (k_{e'}^{UTXO}+1)} \nonumber\\
& \mathbb{P}(A_{i}=a | k^{UTXO}, E_{t}) = \frac{\mathbbm{1}(a\in A(E_{t})) \, (k_{a}^{UTXO}+1)}{\sum_{a' \in \mathscr{A}, a\in A(E_{t})} (k_{a'}^{UTXO}+1)} \nonumber \\
& \mathbb{P}(E_{o}=e | k^{out}) = \frac{k_{e}^{out}+1}{\sum_{e' \in \mathscr{A}} (k_{e'}^{out}+1)} \label{eq:entity}\\
& \mathbb{P}(A_{o}=a | k^{out}, E_{o}) = \mathbbm{1}(o \leq N_t; a = a_{0})\nonumber\\
& \quad \quad \quad \quad \quad + \mathbbm{1} ( o > N_t) \frac{\mathbbm{1}(a\in A(E_{o})) \, (k_{a}^{out}+1)}{\sum_{a' \in \mathscr{A}, a\in A(E_{o})} (k_{a'}^{out}+1)}.\nonumber
\end{align}
\end{definition}
This dependency structure intending to capture the behavior of distinct categories of entities is illustrated in Figure~\ref{fig:plate_notation_simple}.
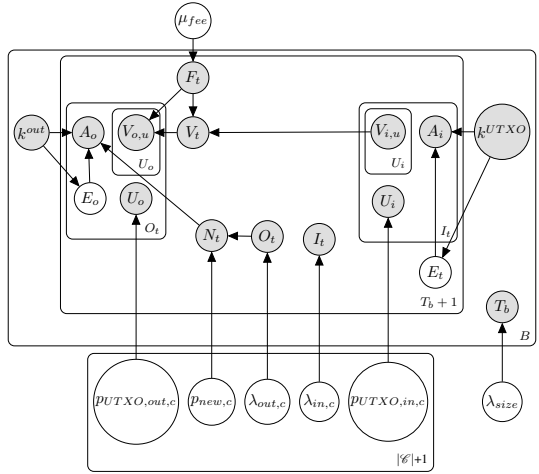
\begin{figure} [!htb]
    \centering
    
\scalebox{0.60}{\begin{tikzpicture}
    
    \node[latent]           (p_utxo_out)        {$p_{UTXO,out,c}$};%
    \node[latent, right=0.2cm of p_utxo_out]           (p_new_c)        {$p_{new,c}$};%
    \node[latent, right=0.2cm of p_new_c]           (lam_out_c)        {$\lambda_{out,c}$};%
    \node[latent, right=0.2cm of lam_out_c]           (lam_in_c)        {$\lambda_{in,c}$};%
    \node[latent, right=0.2cm of lam_in_c]           (p_utxo_in)        {$p_{UTXO,in,c}$};%

    \node[obs, above=3.2cm of p_utxo_out]           (U_o)        {$U_{o}$};%
    \node[obs, above=0.7cm of U_o]              (V_o)      {$V_{o,u}$};%
    \node[obs, left=0.3cm of V_o]           (A_o)        {$A_{o}$};%
    \node[latent, left=0.3cm of U_o]           (E_o)        {$E_{o}$};%
    
    \node[obs, above=2.8cm of lam_in_c]           (I_t)        {$I_{t}$};%
    \node[obs, above=2.8cm of lam_out_c]           (O_t)        {$O_{t}$};%
    \node[obs, above=2.8cm of p_new_c]           (N_new)        {$N_{t}$};%

    \node[obs, right=0.5cm of V_o]           (V_t)        {$V_{t}$};%
    \node[obs, above=0.5cm of V_t]           (F_t)        {$F_{t}$};%
    \node[latent, above=0.5cm of F_t]           (p_fee)        {$\mu_{fee}$};%
    
    \node[obs, above=3.2cm of p_utxo_in]           (U_i)        {$U_{i}$};%
    \node[obs, above=0.8cm of U_i]              (V_i)      {$V_{i,u}$};%
    \node[obs, right=0.3cm of V_i]           (A_i)        {$A_{i}$};%
    \node[latent, below=2.4cm of A_i]           (E_i)        {$E_{t}$};%

    \node[latent, right=1.2cm of p_utxo_in]           (lam_size)        {$\lambda_{size}$};%

    \node[obs, left=0.5cm of A_o]           (d_out_deg)        {$k^{out}$};%
    \node[obs, right=0.5cm of A_i]           (d_utxo_available)        {$k^{UTXO}$};%
    \node[obs, above=1.3cm of lam_size]           (T_b)        {$T_{b}$};%

    \draw[->] (d_out_deg) edge (A_o) %
              (d_out_deg) edge (E_o) %
              (d_utxo_available) edge (E_i) %
              (d_utxo_available) edge (A_i) %
              (lam_size) edge (T_b) %
              (lam_in_c) edge (I_t) %
              (E_o) edge (A_o) %
              (p_utxo_in) edge (U_i) %
              (p_utxo_out) edge (U_o) %
              (lam_out_c) edge (O_t) %
              (O_t) edge (N_new) %
              (p_new_c) edge (N_new) %
              (E_i) edge (A_i) %
              (p_fee) edge (F_t) %
              (F_t) edge (V_o) %
              (N_new) edge (A_o) %
              (F_t) edge (V_t) %
              (V_t) edge (V_o) %
              (V_i) edge (V_t);

    \plate{category}{
        (lam_in_c)(lam_out_c)(p_utxo_in)(p_utxo_out)(p_new_c)
    } {$|\mathscr{C}|$+1};
    
    \plate {UTXO_out} { 
        (V_o)
    } {$U_{o}$}; %
    \plate {out} { 
     (UTXO_out) %
     (U_o)(A_o)(E_o)
    } {$O_{t}$}; %

    \plate {UTXO_in} { 
        (V_i)
    } {$U_{i}$}; %
    \plate {in} { 
     (UTXO_in) %
     (U_i)(A_i)
    } {$I_{t}$}; %

    \plate {Transaction} { 
     (out) %
     (in) %
     (E_i)(O_t)(N_new)(I_t)(F_t)(V_t) %
    } {$T_{b}+1$}; %
    
    \plate {Block} { 
     (Transaction) %
     (T_b)(d_out_deg)(d_utxo_available)%
    } {$ B$}; %
\end{tikzpicture}}
    \caption{\textbf{Block-transaction entity-address} model, including category-specific variables, as well as hidden entities. As per the protocol, all input addresses are associated with only $1$ entity, while output addresses are generally associated with multiple entities.}
    \label{fig:plate_notation_simple}
\end{figure}
\end{subsection}
\begin{subsection}{Model inference}
We assume a known dependency structure and estimate the model parameters. Since the prior is decomposable over nodes, and since all variables are observed in the BT-A model, the MLE inference amounts to local computation over each node and its parents. 

Regarding the BT-EA model, while the hidden entity variables make the inference more complex in general, here we assume that a separate heuristic such as the multi-input heuristic~\cite{AutomaticBtcClustering} allows associating each address with an entity, hence the inference process over the labeled set reduces to the scalable process used for the BT-A model.
\end{subsection}
\end{section}
~~~
\begin{section}{Block-level privacy analysis}\label{sec:privacyAnalysis}
In this section we present an analysis of address re-use behavior in the context of the probabilistic model introduced in the previous section, as well as implications of these results for Bitcoin transaction anonymity.
\begin{subsection}{Attacker model}
We model an attacker, attempting to identify the full set of addresses $A(e)$ associated with an entity $e$. We assume that the attacker uses the standard multi-input heuristic~\cite{AutomaticBtcClustering}, which associates the full set of address inputs for each transaction to a single entity and applies transitive closure. From the perspective of the external attacker, the true set of addresses $A(e)$ of an entity $e$ is partitioned into $A_{e}$ aliases, a-priori seen as distinct entities;
\begin{equation*}
A(e) = \bigcup_{1 \leq i \leq A_{e}} A(e)_{i},
\end{equation*}
where $A(e)_{i}$ denotes the address set associated with alias $i$ of entity $e$. In this setting, when participating in a transaction $t$ on the input side, we consider that the targeted entity $e$ selects $\{  N_{in,i}  \}_{1 \leq i \leq A_{e} }$ addresses from its available set following a generic multinomial distribution with parameters $\{p_{i}\}_{1 \leq i \leq A_{e} }$, which includes the special case for which the alias distribution is a linear function of $k^{UTXO}$. 

This models the typical Bitcoin user who, while being concerned by his privacy, is not particularly careful about  address selection, and uses multiple distinct aliases with distinct address sets, but sometimes mixes these address in the same transaction input, leading to a privacy collapse. 

Given the multi-input heuristic, it is indeed sufficient for an attacker to observe two addresses from distinct aliases and to associate these two aliases to the same entity using the multi-input heuristic. Formally, upon observing the input addresses from a transaction $t$ associated with input entity $e$, the attacker is able to associate the following address set with entity $e$:

{\small{
\begin{equation*}
\{a \in A(e)_{i}\cup A(e)_{j} | \mathbbm{1}(N_{in,i} > 0; \, N_{in,j} > 0), 1 \leq i,j \leq A_{e} \}.
\end{equation*}}}
\end{subsection}
\begin{subsection}{Privacy analysis}
In the following for simplicity we consider a one-step iteration and assume that the attacker is only aware of the set of addresses associated with alias $A(e)_{1}$. In this sense the control parameter $p_{1}$ plays the role of 1-$p_{new}$ from the BT-EA model. We analyze the number $D_{e,t}$ of addresses from entity $e$ that the attacker is able to discover after seeing the addresses involved in $1$ transaction, expressed as:
\begin{equation*}
D_{e,t} = \sum_{i = 2}^{A_{e}} |A(e)_{i}| \mathbbm{1}(N_{in,i} > 0; \, N_{in,1} > 0).
\end{equation*}
We can express the number of discovered addresses $D_{e,t}$ as a function of the alias addresses selection probabilities $p_{i}$.
\begin{prop}[Privacy loss from address re-use]

{\small{
\begin{equation}
\mathbb{E}[D_{e,t}] = \frac{1 - \exp{(-\lambda_{in,c}p_{1})}}{1 - \exp{(-\lambda_{in,c})}}  \sum_{i = 2}^{A_{e}} |A(e)_{i}| (1 - \exp{(-\lambda_{in,c}p_{i})})
\label{eq:prop}
\end{equation}}}
\end{prop}

\begin{proof}
By definition of $D_e,t$, we have
\begin{equation*}
   \mathbb{E}[D_{e,t}] = \sum_{i = 2}^{A_{e}} |A(e)_{i}| \mathbb{P}(N_{in,i} > 0 \land N_{in,1} > 0 | E_{t} = e).
\end{equation*}
Let $B$ be the second factor in the summation term, by marginalizing over $I_{t}$ and using the chain rule, we can write:
\begin{align*}
B & = \mathbb{P}(N_{in,i} > 0 \land N_{in,1} > 0 | E_{t} = e)\\
& = \sum_{n \geq 1} \mathbb{P}(N_{in,i} > 0 \land N_{in,1} > 0 | E_{t} = e, I_{t} = n) \\
& \quad \quad \quad \quad \quad \quad \quad \quad \quad \quad \quad \quad \quad \quad \quad \quad \mathbb{P}(I_{t} = n | E_{t} = e).
\end{align*}
Letting $C$ denote the first factor in the summation term above, we have:

{\small{
\begin{align*}
C & = \mathbb{P}(N_{in,i} > 0 \land N_{in,1} > 0 | E_{t} = e, I_{t} = n)\\
& = \sum_{n_{i} > 0, n_{1} > 0, \sum_{j = 1}^{A_{e}} n_{j} = n} \mathbb{P}( \{ N_{in,j} = n_{j}\}_{j} | E_{t}=e, I_{t}=n)\\
& = \sum_{n_{i} > 0, n_{1} > 0, \sum_{j = 1}^{A_{e}} n_{j} = n}  \frac{n!}{\prod_{j=1}^{A_{e}}n_{j}!}  \prod_{j=1}^{A_{e}} p_{j}^{n_{j}}
\end{align*}}}

where the last equality is obtained by definition of the multinomial distribution. Similarly since the number of input addresses $I_{t}$ follows a binomial distribution we have:
\begin{equation*}
\mathbb{P}(I_{t} = n | E_{t} = e) = \frac{\lambda_{in,c}^{n}}{n!(\exp{(\lambda_{in,c})}-1)},
\end{equation*}
and combining this expression with the expression of $C$, we can simplify the expression of $B$ to finally obtain equation~\eqref{eq:prop}, which concludes the proof.
\end{proof}
With $p_{1}$ as the control parameter, the expression states that the attacker information gain   is an exponential function of the probability of using addresses already identified (i.e. address re-use). The asymptotic behavior of a privacy-conscious user is described next.

\begin{prop}[Privacy-conscious asymptotics]
If $p_{1} \ll p{i}$ we have:
\begin{equation*}
\mathbb{E}[D_{e,t}]  \sim \lambda_{in} \, p_{1} ( |A(e)| - |A(e)_{1}| ).
\end{equation*}
\end{prop}
This result shows that the one-step information gain from the attacker is a linear function of the probability of using already-used addresses, and also linear in the number of addresses typically used as input. This result at the transaction level can be readily extended to a chain-length estimate by accounting for the probability of an entity to transact, as provided explicitly in equation~\eqref{eq:entity} of the BT-EA model. We also highlight that while a low $p_{1}$ models a privacy-conscious user, the user strategy is non-adaptive, in the sense that the user does not try to adjust his strategy based on the attacker strategy.
\end{subsection}
\end{section}
\begin{section}{Probabilistic transaction graph model}
We now consider the behavior of entities across transactions, and assume that entity categories exhibit different behaviors. Given the lack of a-priori underlying modeling structure to this behavior, and given the combinatorial nature of such behavior, we propose a discriminative framework in which model selection can be carried out more efficiently based on a possibly large set of relevant features. We rely on the classical multi-input heuristic~\cite{AutomaticBtcClustering} for defining entities, and formulate a decision-tree based classification problem in the following feature space.

\begin{subsection}{Feature space}
We consider the following five feature classes, and for continuous features explicitly consider the feature mean and standard deviation; address features, entity features, temporal features, graph centrality metric features, motif features.

Address-specific features include attributes such as the total BTC received, the total BTC balance, the number of input/output transactions, etc. Analogous features are defined at the entity level as well as the number and proportion of Coinbase transactions (indicative of BTC creations).

Temporal features are those such as the number of  weeks, months, years of activity. the number of entity traded with per week, month, year, the number of receiving/sending/receiving  sending days, the activity period duration, and the active day ratio.

Motif features are presented in Figure~\ref{fig:2motiffeature}. Here we consider $1$, $2$, and $3$ motifs, extending the $2$-motifs from~\cite{ExchangeAddressClustering}. Motifs are a comprehensive description of the transaction-to-transaction properties.
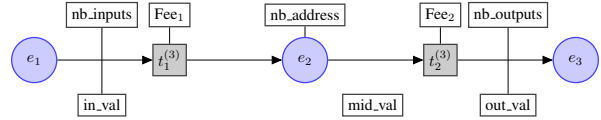
\begin{figure}[!htb]
    \centering
    \scalebox{0.60}{\begin{tikzpicture}[node distance=2cm,bend angle=30,auto]

    \tikzstyle{place}=[circle,thick,draw=blue!75,fill=blue!20,minimum size=10mm]
    \tikzstyle{transition}=[rectangle,thick,draw=black!75,fill=black!20,minimum size=4mm]
    \tikzstyle{label}=[rectangle,thick,draw=black!75,minimum size=4mm]
        \begin{scope}    
    \node [place] (n1) at (0,0) {$e_{1}$};
    \node [place] (n2) at (6,0) {$e_{2}$};
    \node [place] (n3) at (12,0) {$e_{3}$};
    \node [transition] (e1) at (3,0) {$t^{(3)}_{1}$};
    \node [transition] (e2) at (9,0) {$t^{(3)}_{2}$};
    
    \node[label] (l1) at (1.5,1) {nb\_inputs};
    \node[label] (l2) at (10.5,1) {nb\_outputs};
    \node[label] (l3) at (1.5,-1) {in\_val};
    \node[label] (l4) at (10.5,-1) {out\_val};
    \node[label] (l5) at (6,1) {nb\_address};
    
    \node[label] (l6) at (3,1) {Fee$_{1}$};
    \node[label] (l7) at (9,1) {Fee$_{2}$};

    \node[label] (l8) at (7.5,-1) {mid\_val};
    
    \draw[-] (l1) edge (l3) (l2) edge (l4)  (l6) edge (e1)  (l7) edge (e2) (l5) edge (n2);
      
    \draw[->] (n1) edge (e1) (e1) edge (n2) (n2) edge (e2) (e2) edge (n3) ;
    
        \end{scope}
 
    \end{tikzpicture}}
    \caption{\textbf{2-motif features} (rectangular white boxes) annotated over a 2-motif. A N-motif is a path of length $2N$ on the bipartite entity-transaction graph. We distinguish Direct motifs from Loop motifs, the latter indicating that an entity is transacting with itself using distinct addresses.}
    \label{fig:2motiffeature}
\end{figure}
\end{subsection}
\end{section}

\begin{section}{Experiments}\label{sec:numRes}

In this section we present numerical results of our probabilistic Bitcoin Blockchain model. We first describe the training procedure for the generative block model and discuss obtained model parameters. We then turn to the transaction-to-transaction discriminative model results and analyze the properties revealed by the joint analysis.

\begin{subsection}{Dataset}
We consider the set of blocks of height inferior or equal to $514.971$, corresponding to blocks created before March 24th 2018, 15:19:02, which contains about $500.000.000$ addresses. Address labels, revealing entity identifiers, are obtained from WalletExplorer \url{https://www.walletexplorer.com/}. The set of address entity label pairs used has been made available at \url{https://github.com/Maru92/EntityAddressBitcoin}.

We interact with the Blockchain via the BlockSci toolbox v.0.4.5 released on March 16th 2018~\cite{BlockSci}, on a 64 GB machine. The final labeled dataset used in numerical experiments consists of $28.353.493$ addresses, associated with $|\mathscr{E}_{known}|=260$ entities representing $4$ entity categories in the following proportions:
\begin{itemize}
    \item \textit{Exchange} (E): 108 entities, 7.892.587 addresses,
    \item \textit{Service} (S): 68 entities, 17.606.608 addresses,
    \item \textit{Gambling} (G): 65 entities, 2.775.810 addresses,
    \item \textit{Mining Pool} (M): 19 entities, 78.488 addresses.
\end{itemize}
When training the probabilistic model, we restrict ourself to the period from January 1st 2016 to March 16th 2018, where overall patterns are relatively stationary. Indeed since the proposed model is static we do not attempt to study its ability to model transient regimes. We observe $UTXO$ statistics in Table \ref{table:utxoStats} and $UTXO$ distribution in Figure~\ref{fig:utxoScales}, showing wide variability across multiple scales.
\begin{table}[!htb]
\centering
\begin{tabular}{cccccc} 
 \toprule
Quantity & E & S & G & M & All\\
 \midrule
mean $\mu(V_{u,o})$ & 8.62 & 0.53 & 0.11 & 1.27 & 4.39 \\
std $\sigma(V_{u,o})$ & 93.1 & 41.6 & 0.81 & 4.25 & 70.0 \\
 \bottomrule
\end{tabular}
\caption{\textbf{UTXO empirical statistics in BTC:} the UTXO output values have a large standard deviation compared to their mean, and vary significantly across entity categories.}
\label{table:utxoStats}
\end{table}
\begin{figure}[!htb]
    \centering
    \includegraphics[scale=0.5]{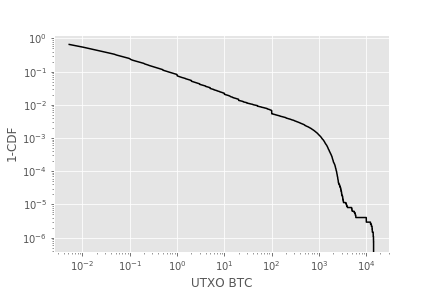}
    \caption{\textbf{BTC UTXO distribution :} $1$ minus the cdf of the UTXO represented in log log coordinates, with $99.9\%$ of the distribution qualitatively following a power law on the interval $[10^{-3},10^3]$.}
    \label{fig:utxoScales}
\end{figure}
\end{subsection}
\begin{subsection}{Transaction subset modeling}
Since we consider a subset of the transaction graph, we  need to model transactions originating from our subset and directed outside it, or vice-versa. We follow the proposed model structure and model the number of external output addresses as a Poisson distribution $\mathscr{P}(\lambda_{sub})$. Transactions from unknown addresses towards known input addresses are modeled with no known input and a number of transactions per block $T_{b,incoming}$ following a Poisson distribution $\mathscr{P}(\lambda_{size,sub})$. Coinbase transactions are created in a similar manner: no inputs, number of addresses in the outputs drawn following a Poisson distribution of parameter $\lambda_{out,sub}$, with new addresses, $p_{new,sub}$, and several UTXOs created per addresses, $p_{UTXO,out,sub}$.
\end{subsection}

\begin{subsection}{Block model training}

We train the model using data from the period January 1st 2016 to March 16th 2018 consisting of about $10$ million addresses. We first verify the main independence assumption, between the number of input addresses and the number of output addresses. Since $\rho_{pearson}(I_{t},O_{t}) = 0.015$, we consider the marginal independence hypothesis validated.

The inference produces a value $\lambda_{size}=65.6$ for both models. In Table~\ref{table:genmodelparams} we present the model parameter results from the model training for the BT-A and BT-EA models.
\begin{table}[!htb]
\centering
\begin{tabular}{cccccc} 
 \toprule
 Parameter& \multicolumn{4}{c}{BT-EA}&\multicolumn{1}{c}{BT-A}\\
 \cline{2-6}
& E & S & G & M & All\\
 \midrule
 $P(E_{t}=e)$ & 0.33 & 0.55 & 0.09 & 0.03 & 1 \\
 $\lambda_{in}$ & 3.79  &  2.58  &  1.98  &  21.2 & 2.99\\
 $\lambda_{out}$ & 0.68  &  1.96  &  0.21  &  7.04 & 1.21\\
 $p_{UTXO,in}$ & 0.95  &  0.92  &  0.84  &  0.67 & 0.92\\
 $p_{UTXO,out}$ & 1.00  &  1.00  &  1.00  &  1.00 & 1.00\\
 $p_{new}$ & 0.23  &  0.20  &  0.47  &  0.55 & 0.26\\
 \bottomrule
\end{tabular}
\caption{\textbf{Model parameters} from calibration on the period from January 1st 2016 to March 16th 2018, for the Exchange, Services, Gambling, Mining Pool categories.}
\label{table:genmodelparams}
\end{table}

The results reflect the idiosyncratic properties of Bitcoin Blockchain transactions, with for instance the need to gather UTXOs from various addresses, which is illustrated by the fact that $\lambda_{in}>\lambda_{out}$. It is also clear from the UTXO parameters that the input parameters are  more discriminative than the output parameters, which reflect transfers from other parties from the perspective of the entity concerned.

Lastly we observe significant address generation distinctions across entity categories, with Gambling and Mining Pools seemingly more privacy-conscious given their higher probability of generating new addresses. They also transact less frequently, using more input addresses. Detailed impact of entity behavior on privacy properties is analyzed subsequently.
\end{subsection}
\begin{subsection}{Block model testing}
In order to assess the model performance, we now evaluate out-of-sample model accuracy. Starting from scratch, we train the model on $4911$ blocks corresponding to the period from January 1st 2017 to January 31st, 2017, and evaluate the model on $2150$ blocks associated with the period from February 1st, 2017, to February 14th, 2017.
\begin{table}[!htb]
\centering
\begin{tabular}{cccccc} 
 \toprule
 Metric& \multicolumn{4}{c}{BT-EA}&\multicolumn{1}{c}{BT-A}\\
 \cline{2-6}
& E & S & G & M & All\\
 \midrule
MSE & 1.22 &  -0.30 &  -0.02 &  0.06 & 1.12 \\
RMSE & 125 &  53.3 &  1.15 &  5.19 & 90.5 \\
MAE & 15.6 &  0.94 &  0.20 &  2.42 & 7.47 \\
RMAE & 1.82 &  1.74 &  1.86 &  1.93 & 1.69 \\
NRMSE & 1.34 &  1.28 &  1.42 &  1.22 & 1.29 \\
 \bottomrule
\end{tabular}
\caption{\textbf{Error statistics in BTC for UTXO output values $V_{u,o}$:} from the BT-A level overall value, as well as per category from the BT-EA model, for the Mean Signed Error (MSE), Mean Absolute Error (MAE), Root Mean Squared Error (RMSE), Relative Mean Absolute Error (RMAE) and Normalized Root Mean Squared Error (N-RMSE) expressed as RMSE divided by $\sigma(V_{u,o})$.}
\label{table:maeResults}
\end{table}

The results from Table~\ref{table:maeResults} illustrate that given the multi-scale nature of the underlying distributions, the model estimates are relatively close on average, i.e. well within an order of magnitude. Furthermore, the BT-EA model significantly reduces the bias (MSE) as well as the variance (RMSE) for most categories. The Exchange category is the only one for which both bias and variance increase, suggesting a fundamental modeling limitation.

The error terms are relatively large in absolute terms for both models, which is largely explained by the inherent variance in the data, both at the population level and at the class level. Indeed, the bias is low and most of the data variance is explained, with a N-RMSE ranging between $1.22$ and $1.34$.
\end{subsection}
\begin{subsection}{Privacy analysis validation}
Given the calibrated model parameters, we now validate experimentally the theoretical privacy properties of Bitcoin Blockchain transactions expressed by equation~\eqref{eq:prop}. We leverage the generative model and attacker model described above to simulate transaction traces and evaluate the proportion of the addresses that are re-identified for distinct categories, as a function of the number of transactions.
\begin{figure}[!htb]
    \centering
    \includegraphics[scale=0.5]{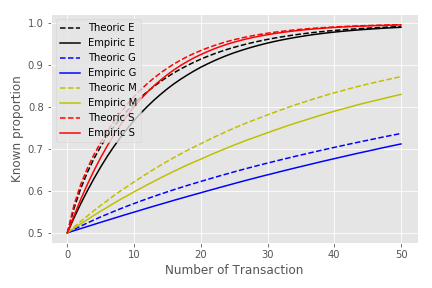}
    \caption{\textbf{Proportion of identified addresses:} by category, as a function of the number of transactions.}
    \label{fig:graphDeAnonymization}
\end{figure}

Figure~\ref{fig:graphDeAnonymization} shows agreement between the analytics results and the simulation of the block model. The figure also illustrates that Exchanges and Services typically are less privacy-conscious (lower probability of generating new addresses, frequent transactions), and hence for an equivalent number of Blockchain transactions, typically reveal a greater proportion of their address set.

Transaction anonymity however depends also on the transaction-to-transaction behavior. Indeed, it is conceivable that certain entities, while not following best block-level practices on address re-use, hence easily identifiable as entities, could be transacting in a way that little information is gathered from their network level transaction structure. In order to assess the latter, we now turn to the numerical results of our proposed network transaction model.
\end{subsection}


\begin{subsection}{Transaction network model}
We use the Python LightGBM implementation of the gradient boosted decision tree model~\cite{NIPS2017_6907} with a $70/30$ training$/$test partition of our dataset. A Gaussian Process (GP)-based optimization procedure for hyper-parameter optimization is implemented using the Python \textit{skopt} library \url{https://scikit-optimize.github.io/} with initial parameter values obtained from a coarse random search. The learning rate  hyper-parameter is optimized over the interval $[0.01,0.5]$ with early stopping after having done a random search over $[0,2]$; the resulting value is  $0.18$. The GP procedure is used with 50 iterations. 

We  make use in total of 10 address features, 8 entity features, 16 temporal features, 42 centrality features, 44 1-motif features, 81 2-motif features, and 114 3-motif features. We present in Table~\ref{table:optimized_gp} the F1, Accuracy and Precision results over the entire dataset and for each category.
\begin{table}[!htb]
\centering
\begin{tabular}{cccc} 
 \toprule
 Category & Accuracy & $F_{1}$ & Precision\\ [0.5ex] 
 \midrule
  Exchange & 0.94  & 0.92   &  0.91 \\
  Gambling & 0.95  &  0.97  &  1.00 \\
  Mining & 0.50  & 0.67   & 1.00  \\
  Service & 0.95  & 0.88   & 0.83  \\
  Overall & 0.92 & 0.91 & 0.92\\
 \bottomrule
\end{tabular}
\caption{\textbf{Classification performance:} for the $4$ entity categories considered, and overall.}
\label{table:optimized_gp}
\end{table}

The results illustrate that the model is able to very well capture the behavior of most entity categories. Furthermore, the network-level privacy analysis confirms the prior block-level analysis, with Mining Pools being the most privacy-conscious. Indeed, considering the most relevant features of the LightGBM model, in a 1 vs. all setting, it appears that for most categories except the Mining Pool, motif features are the most informative, indicating that the LightGBM model is not able to leverage the transaction sub-graph for identification of the Mining Pool category.
\end{subsection}

\end{section}

\begin{section}{Related work}\label{sec:relatedWork}

Heuristics for clustering multiple addresses to an entity have been studied in~\cite{7796940} and consistent address re-use patterns have been shown in~\cite{7816867}. The authors of~\cite{ponzi} attempt to detect Ponzi schemes.

Analysis of the Bitcoin protocol in the context of attacks have been proposed, for instance inference of peer-to-peer communication structure, in~\cite{NIPS2017_6735}, statistical analysis of bloom filters in~\cite{nick2015data}, and analysis of Bitcoin minting patterns in~\cite{mcginn2018toward} with application to de-anonymization. Flow-based address-transaction graph studies can  be found in ~\cite{meiklejohn2013fistful,huang2018tracking,paquet2018ransomware}. The obfuscation of Bitcoin transactions traceability has  been considered in~\cite{priceofanonymity}.

Several studies have applied discriminative models to the problem of de-anonymizing Bitcoin transactions, with for instance the use of transaction-specific features in~\cite{multiClass}, able to achieve $70\%$ accuracy for classifying entities into several types. In ~\cite{ExchangeAddressClustering}, the authors introduce transactions paths with application to the detection of Bitcoin exchanges, and achieve greater than $80\%$ accuracy. Similar transactions paths features are used in~\cite{Jourdan2018} for a 5-class classification problem with above $90\%$
 accuracy results.

\end{section}


\begin{section}{Conclusion}\label{sec:conc}

In this work, we proposed a  probabilistic model of the Bitcoin Blockchain which accounts for the complex Bitcoin protocol features. The model consists of a hierarchical structure from unspent transaction output (UTXO), to address, transaction, and block. We take into account entity modeling, including features relevant for  robustness to de-anonymization attacks, namely address re-use patterns. We also propose a discriminative model of transaction-to-transaction behavior and show its effectiveness in practice.

We analyzed the accuracy of the generative model using a large Bitcoin dataset of more than $10$ million address vertices, discussed the significant block-level heterogeneity of the model parameters across entity categories, and provide a complementary analysis of transaction-to-transaction behavior using the discriminative model. We consider in particular the de-anonymization properties of certain behaviors, which is one of the main focus areas of Bitcoin studies.

Extensions of this work may include the design of more complex graphical models including latent variables for modeling  transaction intent, and shared side-information across entities, inducing multivariate preferential attachment. A significant challenge for such models with more complex dependency structure and hidden variables is the design of a tractable training and inference procedure given the large-scale nature of such public cryptocurrency transaction graphs.
 \end{section}
\newpage
\bibliography{mybib}
\bibliographystyle{plain}
\end{document}